\documentclass[letterpaper]{article} % DO NOT CHANGE THIS
\usepackage{aaai25}  % DO NOT CHANGE THIS
\usepackage{times}  % DO NOT CHANGE THIS
\usepackage{helvet}  % DO NOT CHANGE THIS
\usepackage{courier}  % DO NOT CHANGE THIS
\usepackage[hyphens]{url}  % DO NOT CHANGE THIS
\usepackage{graphicx} % DO NOT CHANGE THIS
\usepackage{natbib}  % DO NOT CHANGE THIS AND DO NOT ADD ANY OPTIONS TO IT
\usepackage{caption} % DO NOT CHANGE THIS AND DO NOT ADD ANY OPTIONS TO IT
\usepackage{algorithm}
\usepackage{booktabs} 

\usepackage{amsmath, amsfonts, amsthm, amssymb}
\usepackage{subcaption}
\usepackage{algpseudocode}

\newcommand{\bbE}{\mathbb{E}}
\newcommand{\bbR}{\mathbb{R}}
\newcommand{\bfQ}{\mathbf{Q}}
\newcommand{\bbRnm}{\mathbb{R}^{n\times m}}
\newcommand{\calA}{\mathcal{A}}

\newcommand{\calO}{\mathcal{O}}

\newcommand{\calS}{\mathcal{S}}

\newcommand{\calT}{\mathcal{T}}
\newcommand{\calQ}{\mathcal{Q}}
\newcommand{\calM}{\mathcal{M}}
\newcommand{\calN}{\mathcal{N}}

\newcommand{\betahat}{\hat{\beta}}
\newcommand{\sumij}{\sum_{i=1}^n\sum_{j=1}^m}
\newcommand{\pialpha}{{\pi^\alpha}}
\DeclareMathOperator*{\argmax}{\arg\!\max}

\usepackage{graphicx}
\usepackage{thmtools}

\usepackage{newfloat}
\usepackage{listings}
\DeclareCaptionStyle{ruled}{labelfont=normalfont,labelsep=colon,strut=off} % DO NOT CHANGE THIS
\floatstyle{ruled}
\newfloat{listing}{tb}{lst}{}
\floatname{listing}{Listing}
\title{Multi Agent Reinforcement Learning \\for Sequential Satellite Assignment Problems}
\author {
    Joshua Holder\textsuperscript{\rm 1}, 
    Natasha Jaques\textsuperscript{\rm 2}, 
    Mehran Mesbahi\textsuperscript{\rm 1}
}
\affiliations {
    \textsuperscript{\rm 1}Department of Aeronautics and Astronautics, University of Washington, Seattle, WA 98195\\
    \textsuperscript{\rm 2}Department of Computer Science, University of Washington, Seattle, WA 98195\\
    josh.holder72@gmail.com, nj@cs.washington.edu, mesbahi@uw.edu
}

\usepackage{bibentry}
\begin{document}

\maketitle

\begin{abstract}
Assignment problems are a classic combinatorial optimization problem in which a group of agents must be assigned to a group of tasks such that maximum utility is achieved while satisfying assignment constraints. Given the utility of each agent completing each task, polynomial-time algorithms exist to solve a single assignment problem in its simplest form. However, in many modern-day applications such as satellite constellations, power grids, and mobile robot scheduling, assignment problems unfold over time, with the utility for a given assignment depending heavily on the state of the system. We apply multi-agent reinforcement learning to this problem, learning the value of assignments by bootstrapping from a known polynomial-time greedy solver and then learning from further experience. We then choose assignments using a distributed optimal assignment mechanism rather than by selecting them directly. We demonstrate that this algorithm is theoretically justified and avoids pitfalls experienced by other RL algorithms in this setting. Finally, we show that our algorithm significantly outperforms other methods in the literature, even while scaling to realistic scenarios with hundreds of agents and tasks.
\end{abstract}

\begin{links}
    \link{Code}{https://github.com/Rainlabuw/rl-enabled-distributed-assignment}
\end{links}

\section{Introduction}

Large-scale distributed systems like the power grid, transportation networks like Uber and Lyft, and satellite internet constellations are increasingly integrated in and critical to every aspect of our day-to-day lives, and will only become more so as time goes on. We can model these systems as a large group of agents working together to achieve broader goals - individual batteries and power plants working to satisfy grid-wide demand \cite{giovanelli2019power}, requested rides being distributed between cars \cite{qin2022reinforcement}, or satellites working together to provide internet across the Earth \cite{lin2022dynamicbeam}. 

In order to operate these systems efficiently, it is often necessarily to solve optimization problems at a massive scale. Given $n$ agents and $m$ tasks, one of the most natural optimization questions to ask is ``How can agents be optimally assigned to tasks?"

% While we will see later that the simplest version of this question admits efficient solutions, most realistic problems are more complex. Specifically, in real systems, assignments must be made repeatedly rather than at a single instant. Furthermore, these problems are often state dependent - when an assignment is made, the state of the system changes, which affects the value of future assignments (i.e. a satellite has to change its orientation to complete a task). These more complex problems, including sequential assignment problems as defined in Equation \ref{eqn:sap_formulation}, are often NP-hard \cite{gerkey2004formal} and thus difficult to approach with classical methods.
While we will see later that the simplest version of this question admits efficient solutions, most realistic problems are more complex. Specifically, in real systems, assignments must be made repeatedly rather than at a single instant. Furthermore, these problems are often state dependent - when an assignment is made, the state of the system changes, which affects the value of future assignments (i.e. a satellite has to change its orientation to complete a task). These more complex problems are often NP-hard \cite{gerkey2004formal} and thus difficult to approach with classical methods.

The temporal nature of this problem suggests that sequential decision-making techniques like reinforcement learning (RL) may be an attractive solution. However, as the number of agents in the environment grows, so too does the complexity of solving the problem with a centralized algorithm \citep{marl-book}. For this reason, we look to multi-agent reinforcement learning (MARL) to enable our solution to scale up to the massive problem size required in realistic problem domains. Naive application of MARL is difficult for several reasons; rewards must be specified such that cooperation on the global objective is guaranteed, and actions taken by individual MARL agents must attempt to satisfy the constraint that each agent is assigned to a unique task (i.e. to avoid conflicting assignments).

\begin{figure*}
    \centering
    \includegraphics[width = 1.8\columnwidth]{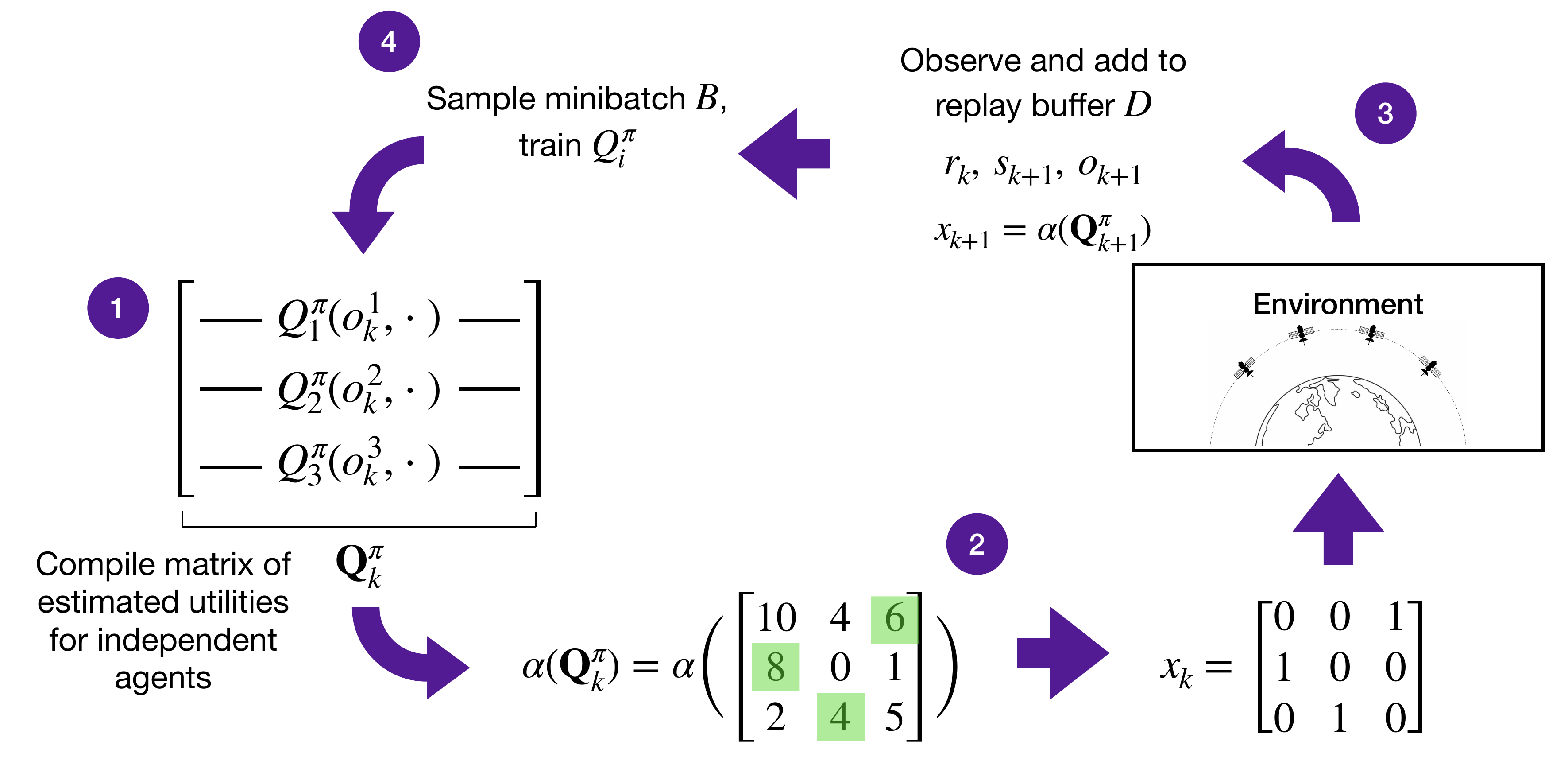}
    \caption{Architecture of REDA. 1) Calculate independent estimates of future utility for each agent, combine into a matrix. 2) Select joint assignment $x_k=\alpha(\bfQ_k^\pi)$ which maximizes social utility, not utility for any given agent. 3) Execute $x_k$ in environment and observe results. 4) Train agents' independent value estimates based on minibatch from replay buffer.}
    \label{fig:architecture}
\end{figure*}

In this work, we present a novel, theoretically justified algorithm that  addresses both of these challenges. Rather than having agents learn to assign themselves to tasks directly, we have each agent learn the expected value of an assignment, and use these learned values as the input to an optimal distributed task assignment mechanism. This allows agents to execute joint assignments that satisfy assignment constraints and avoid selfishness, while learning a joint policy which is near-optimal on the level of the entire system. This architecture is depicted in Figure \ref{fig:architecture}.

While prior work has used MARL in a similar fashion to address assignment problems, this work often focused on variants of the problem specific to ride-sharing \cite{shah2020neural, azagirre2024better}. 

Throughout this work, we will take a far more general approach, and focus instead on satellite internet constellations as a novel motivating example (i.e. agents are satellites, and tasks are regions on Earth to provide internet to). This is a compelling application for a few reasons. First, satellite assignment problems are uniquely complex in that orbital mechanics dictates that a satellite cannot accomplish the same task indefinitely, and thus that frequent transitions between tasks must be considered. Second, these systems have exploded in size in just a few years to contain thousands of agents, with few existing algorithms which can provide the high degree of autonomy, efficiency, and resilience required in this environment. Finally, these systems are incredibly expensive, with the marginal cost of a satellite approaching $\$500${\small,}$000$. If more efficient planning enables reducing the size of the constellation by even a few satellites out of thousands, it could save tens of millions of dollars.

As such, we apply our algorithm to a highly realistic satellite constellation management scenario, complete with high-fidelity orbital mechanics and hundreds of satellites and tasks, a scale which is extremely uncommon among non-heuristic approaches in the literature \cite{9115252}. Despite this, we find that our approach outperforms other state-of-the-art approaches by $20-50\%$ across both multi-agent reinforcement learning (COMA \cite{foerster2018counterfactual}, IQL \cite{matignon2012independent}, IPPO \cite{de2020independent}) and classical optimization (HAAL, \cite{holder2024haal}).

In summary, the contributions of this paper are:
\begin{itemize}
 \item A MARL approach for assignment problems which seamlessly integrates existing greedy planners into MARL, but improves upon their solution for better long-term planning.
 \item Novel insight into the workings of the method through simple experiments, direct comparison with state-of-the-art RL methods, and theoretical analysis that provides intuition on global convergence properties.
 \item Empirical results on a real satellite assignment problem that show a vastly improved ability to manage long-term resource constraints even when planning in complex environments with hundreds of satellites and tasks.
\end{itemize}

\section{Technical Preliminaries}
Mathematically, we can formulate the simplest version of an assignment problem with $n$ agents and $m$ tasks as:
\begin{equation}\label{eqn:standard-assignment-prob}
    \alpha(\beta) = \underset{x\in X}{\argmax} \ \sumij \beta_{ij} x_{ij}
\end{equation}
where: 
\begin{itemize}
    \item $\beta\in\bbR^{n \times m}$ is the \textit{benefit matrix}, where $\beta_{ij}$ corresponds to the utility of agent $i$ completing task $j$.
    \item $x\in X \subset \{0,1\}^{n \times m}$ is the \textit{assignment matrix}, where $x_{ij}=1$ if agent $i$ is assigned to task $j$, and $x_{ij}=0$ otherwise.
    \item $X := \{x \in \{0,1\}^{n\times m} \ | \ \sum_{j=1}^m x_{ij} = 1 \ \forall i, \ \sum_{i=1}^n x_{ij} \leq 1 \ \forall j\}$ is the \textit{set of valid assignments}. This corresponds to the set of assignment matrices such that each agent completes $1$ task, and each task is completed by at most $1$ agent.
\end{itemize}

When benefits $\beta$ are given, this is a well-studied problem for which a solution can be easily obtained in polynomial time (i.e. with a single Python command) \cite{kuhn1955hungarian}. As such, we denote the solution to Equation \ref{eqn:standard-assignment-prob} as a function \newline $\alpha : \bbRnm \to X$.
% Importantly, this solution can also be computed near-optimally in a fully distributed fashion, where an agent $i$ submits a bid on a task $j$ they want to complete based on the benefit $\beta_{ij}$ they expect to receive from completing it, and exchange bids until agents are matched with tasks for which they are willing to pay the highest amount \cite{zavlanos2008distributed}. We denote the distributed solution to (\ref{eqn:standard-assignment-prob}) as $\alpha^D : \bbRnm \to X$.

Consider the sequential assignment problem (SAP), a more complex case where assignments need to be made at several time steps, and the assignment benefits $\hat{\beta}$ depend on some state $s \in S$ that evolves according to a transition function $\calT: S \times X \to S$.
\begin{equation}\label{eqn:sap_formulation}
    \underset{\pi}{\max} \ \bbE^\pi \bigg[ \sum_{k=1}^T \sumij \gamma^{k-1} [\hat{\beta}(s_k)]_{ij} [x_k]_{ij} \bigg]
\end{equation}
where $x_k$ and $s_k$ denote the assignment and state at time step $k$, $s_1\sim S_0$, $\gamma$ is the discount factor, and $\bbE^\pi$ denotes that states evolve according to the transition dynamics $s_{k+1}\sim\calT(s_k, x_k)$ and that assignments are chosen with respect to the policy $x_k \sim \pi(s_k)$.

% There is a clear parallel between this problem and a more standard finite-time Markov Decision Process $\mathcal{M} = (S, \calA, S_0, \calT, r, \gamma)$, where the standard action $a$ and action space $\calA$ correspond to an assignment $x$ and space of valid assignments $X$ in the SAP formulation. Accordingly the reward function is given by $r(s,x)=\sum_{i,j} [\betahat(s)]_{ij}x_{ij}$.
The clear parallel between this problem and a more standard finite-time Markov Decision Process $\mathcal{M} = (S, \calA, S_0, \calT, r, \gamma)$ seen in RL is outlined in Table \ref{table:rl_sap_eqiv}.
\begin{table}[h]
  \caption{Mapping between classic RL and SAP problem formulation.}
  \label{table:rl_sap_eqiv}
  \centering
  \begin{tabular}{lll}
    \toprule
    \cmidrule(r){1-2}
    \textbf{Classic RL Formulation}     & \textbf{SAP Formulation} \\
    \midrule
    $\calA$ (action space) & $X$ (valid assignment space)     \\
    $a\in\calA$ (action)     & $x\in X$ (assignment)      \\
    $r(s,a)$ (reward func.)     & $\sum_{i,j} [\betahat(s)]_{ij}x_{ij}$ (benefit func.)  \\
    \bottomrule
  \end{tabular}
\end{table}

In this setting, the $Q$-function for policy $\pi$ is defined as $Q^\pi(s_k,x_k):= \bbE^\pi [ r(s_k,x_k) + \sum_{t=k+1}^T \gamma^{t-k} r(s_t, x_t)]$, which represents the total expected future reward of beginning in state $s_k$, making assignment $x_k$, and following policy $\pi$ thereafter.
% In this setting, the value function for a policy $\pi$ is defined as $V^\pi(s_k) := \bbE^\pi [ \ \sum_{t=k}^T \gamma^{t-k}r(s_t,x_t) \ ]$, and a $Q$-function for policy $\pi$ is defined as a $Q^\pi(s_k,x_k):= r(s_k,x_k) + \gamma \bbE^\pi V^\pi(s_{k+1})$.

In order to scale solutions to large groups of agents and tasks, it is desirable to formulate this centralized RL problem as a MARL problem. In the MARL case, we define a \textit{joint assignment} $x=(x^1, \cdots, x^n)$ and a \textit{joint policy} $\pi=(\pi^1, \cdots, \pi^n)$, $x^i \sim \pi^i$. The assignment space for a single agent $x^i:=\argmax_j \ x_{ij}\in [m]$ is now a single integer denoting the task agent $i$ is assigned to.

We assume that the environment is partially observable, and that agents are equipped with an observation function $\calO^i:S \to O^i$ which they use to observe components of the state $o^i\sim\calO^i(s)\in O^i$ on which to condition their $Q$-functions.

We now describe the deficiencies in previous approaches to solving the SAP.

\section{Related Work}

\subsection{Classical methods for the sequential assignment problem}
Although efficient, optimal solutions exist for finding an optimal assignment for a single time step (\ref{eqn:standard-assignment-prob}), the SAP (\ref{eqn:sap_formulation}) is NP-hard except in trivial cases and is thus much harder to tackle using classical approaches. Much of the existing work relies on purely heuristic methods \cite{pachler2021static} or ignores the state-dependent aspect of the problem entirely \cite{vanphuc2022jointbeam}. In one recent work \cite{holder2024haal}, the authors develop HAAL, which uses information from several time steps to generate assignments in the style of model-predictive control, but the method is limited to a specific class of deterministic SAPs.

\subsection{RL for the sequential assignment problem}
% A recent line of work \cite{qin2022reinforcement} applies a similar MARL formulation to the driver-rider matching problem. In \cite{azagirre2024better}However, these works do not address the general SAP, instead focusing on highly specific transportation-domain applications. Our work presents a generalized method for solving SAPs, and includes a novel strategy combining ILP
% A line of recent work applies an approach similar to ours to the DiDi and Lyft transportation networks \cite{tang2019deep, azagirre2024better}. However, our method bootstraps policies using an optimal assignment mechanism, allowing us to use neural networks rather than linear function approximation. Additionally, we employ a novel formulation to expose theoretical insights into the workings of the method. 

Similar to \cite{shah2020neural}, we bootstrap RL learning from a greedy assignment mechanism and use Gaussian noise to induce exploration. However, Shah et al is focused on a particular variant of the SAP specific to the problem of pooling rides when ride-sharing. By contrast, we present a generalized method for solving the SAP across domains, and present both a theoretical analysis and an empirical comparison with state-of-the-art RL and MARL algorithms, providing intuition into the global convergence properties of the method.
% A recent line of work \cite{shah2020neural, qin2022reinforcement, azagirre2024better} applies a similar MARL formulation to the driver-rider matching problem. However, these works do not address the general SAP, instead focusing on highly specific transportation-domain applications. Our work presents a generalized method for solving SAPs across domains, and leverages existing optimization solvers and a novel exploration method to enable stable and sample-efficient training with neural networks. Additionally, through theoretical analysis and direct comparison to existing RL algorithms, we provide novel intuition into the global convergence properties of the method.

Other work at the intersection of RL and assignment problems allows agents to directly learn to propose their value for completing a task \cite{chang2020decentralized}, introducing problems of incentive-compatibility, or do not directly learn the values of individual assignments \cite{hwang2007soccer}, greatly limiting the expressiveness of the method. By contrast, our method avoids such problems by using a centralized consensus mechanism, while maintaining the expressiveness of learned values.
% Limited work exists applying RL directly to the SAP, despite wide interest in foundational and applied aspects of MARL~\cite{Bertsekas2020-ls,Bertsekas2021-dm,Zhang2021-jo}. 
% In \cite{hwang2007soccer}, the authors use RL learn to switch between manually-specified high-level assignment strategies and use $\alpha$ to make assignments in a given strategy. However, they do not directly learn the values of individual assignments, greatly limiting the expressiveness of the method.
% In \cite{chang2020decentralized}, the authors pose the RL problem as a form of assignment problem. However, they do not consider a multi-agent setting, and by allowing agents to directly learn to propose their value for completing a task, introduce problems of incentive-compatibility \cite{roughgarden2010algorithmic}. 

\subsection{Why don't normal MARL methods work?}
Existing MARL techniques have seen success when applied to incredibly complex problems like DOTA 2 \cite{berner2019dota}, but the SAP poses a unique challenge.

The first challenge is deciding how to specify the reward function for an agent $i$, $r_i(s,x)$. Given that our objective is a global maximization over tasks completed by \textit{all} agents rather than any single agent, one might be tempted to use cooperative rewards, where $r_i(s,x)=r(s,x)$, as in \cite{rashid2020monotonic, sunehag2017value}. However, as we will see later, agents with this reward function struggle to disentangle the effect of their actions among the many other agents in the group, even when applying techniques like COMA, which are designed to enable a single agent to assess its counterfactual impact on the joint reward \cite{foerster2018counterfactual}.

Conversely, one might take note of the recent success of completely independent agents in cooperative domains \cite{de2020independent, papoudakis2020benchmarking, yu2022surprising} and provide agents with the rewards they yield solely from tasks they complete, $r_i(s,x^i) = [\hat{\beta}(s)]_{ix^i}$. When specifying rewards in this way, the learning problem is significantly easier because rewards correlate more directly with agent actions, and we hope that cooperation between agents in allocating tasks emerges naturally from the training process. Empirically, however, we see this behavior does not emerge easily, and that many agents simultaneously assign themselves to the most valuable tasks, even if the tasks can only be completed by a single agent.

This leads us to the second difficulty of applying MARL to assignment problems, which is that without an explicit constraint, joint assignments often consist of multiple agents completing the same task, meaning that $x\notin X$. Clearly, some centralization or communication between agents is necessary to ensure that agents do not duplicate assignments. Our method provides a principled approach for allowing agents to resolve conflicts and learn to make socially optimal and valid assignments $x\in X$.

\section{Method}\label{section:method}
The function $\alpha$ introduced in Equation \ref{eqn:standard-assignment-prob} is a solution to both of these difficulties. First, note that given information about \textit{the benefit to individual agents} $i$ for completing tasks $j$, $\beta_{ij}$, $\alpha$ yields assignments which are optimal \textit{on the level of the entire group}. Additionally, $\alpha$ yields joint assignments $x\in X$---i.e. those that avoid assigning multiple agents to the same task---by definition.

Thus, we would like to use $\alpha$ as our joint policy, $\pi(s)=\alpha(\beta)$. However, in the SAP setting the state-dependent nature of the planning problem means that we do not know the long-term benefit of assignments \textit{a priori}, and thus cannot easily access a $\beta\in\bbRnm$. Instead, we propose to \textit{learn} the expected future value of each assignment $i\to j$ from experience, and use these values as input to $\alpha$.

% We now describe our method, which shares many structural similarities with Independent $Q$-Learning \cite{tan1993multi} in that it consists of individual, decentralized agents with a shared network which attempts to learn $Q$ values for agent assignments. 
% We have seen in (\ref{eqn:standard-assignment-prob}) that given knowledge about the value $\beta_{ij}$ of every agent $i$ completing every task $j$, $\alpha(\beta)$ provides a method of calculating a globally optimal, valid joint assignment for the group.
% First, we describe how $Q$-functions are used in combination with $\alpha$.

\subsection{Using agent Q-functions in the assignment mechanism $\alpha$}
Many $Q$-learning based algorithms such as DQN \cite{mnih2013playing} require the ability to take actions $\epsilon$-greedily with respect to the current policy, $x_{k}=\argmax_{x^*\in X} Q^\pi(s_k, x^*)$. However, in the assignment setting, acting greedily with respect to the joint policy $\pi$ becomes a difficult non-convex optimization problem. Our method provides a way of approximating this behavior.

A key part of why we can do this is because of a decomposition of the joint $Q$ function that exists in the assignment problem setting (similar to the one used in \cite{sunehag2017value}), which we outline in Theorem \ref{theorem:q_decomp}.
\begin{restatable}[Decomposition of $Q^\pialpha$ into $Q_i^\pialpha$]{thm}{decompthm}\label{theorem:q_decomp}
    Let $\pi^\alpha: S \to X$ be a constant, deterministic joint policy. Define the $Q$-function for an individual agent with respect to this joint policy $\pi^\alpha$ as:
    \begin{equation*}
        Q^\pialpha_i(s_k, j) := \bbE^\pialpha \bigg[ r_i(s_k, j) + \sum_{t=k+1}^T \gamma^{t-k} r_i(s_t, x_t^i)\bigg]
    \end{equation*}

    Then, in the assignment problem setting, where $r(s,x)=\sumij\hat{\beta}(s)x_{ij} = \sum_{i=1}^n r_i(s,x^i)$,
    \begin{equation}\label{eqn:q_decomp}
        Q^\pialpha(s_k, x) = \sumij Q_i^\pialpha(s_k, x^i)x_{ij} \quad \text{for } \; x=\pi^\alpha(s_k).
    \end{equation}
\end{restatable}
\begin{proof}
    This follows from the fact that $r(s,x)=\sum_{i=1}^n r_i(s,x^i)$---that is, the sum of the reward for the individual agents is equal to the joint reward, and the reward for agent $i$ is conditioned only on agent $i$'s assignment, $x^i$. See the supplemental materials for the complete proof.
    %the definition of $Q_i^\pialpha$ with respect to this joint policy $\pi$ rather than agent policies $\pi^i$ as is standard, as well as
\end{proof}

In words, given a state $s_k$, $Q^\pialpha_i(s_k, j)$ is defined as the total expected future reward that agent $i$ will obtain, given that agent $i$ is assigned to task $j$, and then that all agents follow the joint policy $\pialpha$ for future assignments. Given the partial observability of our environment, in practice we make the approximation $Q_i^\pialpha(s,j)\approx Q^\pialpha_i(\calO^i(s),j)$.

Here, we can begin to see the clear connection between Equations \ref{eqn:standard-assignment-prob} and \ref{eqn:q_decomp}; if we define our benefit matrix (previously $\beta$) to be $\bfQ^\pi_k\in\bbRnm$ such that $[\bfQ^\pialpha_k]_{ij} = Q^\pialpha_i(o^i_k,j)$, then $\argmax_{x^*\in X} Q^\pialpha(s_k, x^*) \approx \alpha(\bfQ^\pi_k)$. Then, we can make assignments according to:
\begin{equation}\label{eqn:reda_policy}
    x_k=\alpha(\bfQ^{\pialpha}_k)\approx\argmax_{x^*\in X}Q^{\pialpha}(s_k, x^*).
\end{equation}

While Equation \ref{eqn:q_decomp} only holds for $x=\pialpha(s)$, if we assume that policies change slowly during the learning process such that $\alpha(\bfQ^\pialpha_k)\approx \pialpha(s)$, then Equation \ref{eqn:reda_policy} is a valid approximation. Thus, when we learn estimates of $Q^\pialpha_i(o^i_k, j)$ directly from experience, we can build $\bfQ_k^\pialpha$ and have agents act not by picking assignments that are best for themselves, but through the mechanism $\alpha$ which is guaranteed to return a socially optimal outcome for the group. This motivates our algorithm which we fully specify in the next section.

\subsection{RL-Enabled Distributed Assignment (REDA)}
REDA, described in Algorithm \ref{alg:REDA} and depicted in Figure \ref{fig:architecture}, is our method for generating solutions to the SAP, with key differences from a standard independent DQN algorithm (e.g. \cite{tampuu2017multiagent}) on lines \ref{algline:jumpstart_explore}-\ref{algline:explore_end} and \ref{algline:build_target_start}-\ref{algline:build_target_end}.

\begin{algorithm}[ht]
    \caption{RL-Enabled Distributed Assignment (REDA)}\label{alg:REDA}
    \textbf{Given}: state-dependent benefit function $\hat{\beta}:\calS \to \bbRnm$ 
    \begin{algorithmic}[1]
        \State Initialize $Q$-network parameters $\theta$, target $Q$-network parameters $\bar{\theta}=\theta$
        \State Initialize a replay buffer $D$
        \For{episode $e=1, 2, ...$}
            \For{time step $k=1,...,T$}
                \State Collect joint observation $o_{k}=(o_{k}^1, \cdots, o_{k}^n)$
                \State With probability $\epsilon$: $x_k \gets \alpha(\hat{\beta}(s_k))$ \ \textit{(act greedily w/r/t the current benefit matrix)}\label{algline:jumpstart_explore}
                \State Otherwise:\label{algline:greedy_act}
                \State$\quad$ Build $\bfQ^\pi_k$ such that $[\bfQ^\pi_k]_{ij} \gets Q^\pi_i(o^i_k, j; \theta)$
                \State $\quad$ $\bfQ_{\text{avg}} \gets \frac{1}{nm}\sumij |[\bfQ^\pi_k]_{ij}|$
                \State $\quad$ Generate perturbation matrix $\boldsymbol{\xi}\in\bbRnm$, where $\boldsymbol{\xi}_{ij} \sim N(0, 2\bfQ_{\text{avg}}\epsilon)$
                \State $\quad$ $x_k \gets \alpha(\bfQ^\pi_k + \boldsymbol{\xi})$ \ \textit{(act $\sim$optimally w/r/t the estimated values of $Q^\pi_i(o^i_k, j; \theta)$)}\label{algline:explore_end}
                \State Collect joint assignment $x_k=(x_k^1, \cdots, x_k^n)$
                \State Collect joint reward $r_k=(r_k^1, \cdots, r_k^n)$,\newline where $r^i_k \gets r_i(s_k, x^i_k)$
                \State Observe next state $s_{k+1}\sim\calT(s_k, x_k)$ \State Collect joint observation $o_{k+1}$
                \State Store joint transition $(o_k, x_k, r_k, o_{k+1})$ in $D$
                \State Sample random mini-batch of $B$ joint transitions $(o_t, x_t, r_t, o_{t+1})$ from $D$
                \If{$s_{t+1}$ is terminal}
                    \State Targets $y_t^i \gets r_t^i$ for all $i$
                \Else
                    \State Build $\bfQ^\pi_{t+1}$, $[\bfQ^\pi_{t+1}]_{ij} \gets Q^\pi_i(o_{t+1}^i, j; \theta)$\label{algline:build_target_start}
                    \State $x_{t+1} \gets \alpha(\bfQ^\pi_{t+1})$
                    \State Targets $y_t^i \gets r_t^i + \gamma Q_i^\pi(o_{t+1}^i, x_{t+1}^i; \bar{\theta}) \ \forall i$\label{algline:build_target_end}
                \EndIf
                \State Loss $\mathcal{L}(\theta) \gets \frac{1}{B}\sum\limits_{t=1}^B \sum_{i=1}^n \bigg( y_t^i - Q(o_t^i, x_t^i; \theta) \bigg)^2$
                \State Update parameters $\theta$ by minimizing $\mathcal{L}(\theta)$
                \State Update target network parameters $\bar{\theta}$ periodically
            \EndFor
        \EndFor
    \end{algorithmic}
\end{algorithm}

\newpage
\textbf{Bootstrapping from a greedy policy.} Sequential assignment problems are unique in that there always exists a sub-optimal, non-parametrized policy with which we can bootstrap our policy from; $\pi(s) := \alpha(\hat{\beta}(s))$, where we simply make the greedy assignment with respect to the current benefit matrix at every time step, without regard for future benefits. At the beginning of training, we act with this greedy policy with probability $\epsilon$, filling our replay buffer with reasonable state-assignment pairs before beginning to learn to improve on this policy. 
% This amounts to a pretraining strategy similar to the one developed in \cite{uchendu2023jump} to bootstrap a value-based policy with a non-parametrized guide policy. Note that the guide policy requires access to the full state of the system $s_k$, which means that training must be done in a centralized setting. 

\textbf{Exploration.} To induce further exploration, we also add randomly distributed noise $\boldsymbol{\xi}$ to $\bfQ^\pi$, scaled by the current average magnitude of $\bfQ^\pi$, $\bfQ_{\text{avg}}$, such that sub-optimal joint assignments are selected with some probability. This is a more effective exploration strategy than making entirely random joint assignments $x\in X$ given the size of the search space, $|X|=\frac{m!}{(m-n)!}$.

Unlike previous work \cite{shah2020neural}, perturbations do not need to be tuned according to reward magnitude because the noise is scaled directly according to the values of $\bfQ^\pi$.

\textbf{Target specification.} Another important aspect of REDA is the way learning targets are specified. Because the policy $\pi$ can only select assignments $x \in X$, targets must also satisfy this constraint. In other words, Lines \ref{algline:build_target_start} through \ref{algline:build_target_end} express $y=r+\max_{x^*\in X} Q^\pi(s,x)$ rather than $y=r+\max_{x^*} Q^\pi(s,x)$. We find the best assignment $x^*\in X$ by again using $\alpha$---following the standard DQN paradigm \cite{mnih2013playing}, $\bfQ_{t+1}^\pi$ is generating using the value network with parameters $\theta$, but the assignment $x_{t+1}$ is evaluated using the target network with parameters $\bar{\theta}$.

\textbf{Computing $\alpha(\bfQ^\pi)$ in a distributed way.} In deployment, each agent can independently compute $\alpha(\bfQ^\pi)$, either by receiving the values of $Q_i^\pi(o_k^i,j) \ \forall i, j$, or by using market-based mechanisms in which agents exchange bids with neighboring agents until they are matched with the task for which they are willing to pay more than all other agents, as in \cite{zavlanos2008distributed}. This means that as long as agents can communicate about task values, the algorithm can be executed in a distributed fashion.

\subsection{Theoretical justification}
To further motivate why REDA produces sensible policies, we can show that the REDA target update causes $Q_i$ to converge to the true $Q^\pi_i$ under reasonable assumptions. 
% We can now show that $Q_i^\pi(o_k^i,j)$ can be easily estimated from experience simply by using SARSA-style updates. Formally:
\begin{restatable}{lmma}{contractlemma}\label{lemma:contraction}
    Let $Q_i\in\calQ$ be an arbitrary $Q$-function. Let $F:\calQ\to\calQ$ be the operator corresponding to the REDA target update in the tabular case, without target networks or the greedy guide policy:
    \begin{equation*}
        (FQ_i)(o_k^i,j) = \bbE^\pi \bigg[ r_i(s_k,j) + \gamma Q_i(o_{k+1}^i, x_{k+1}^i)\bigg]
    \end{equation*}
    
    Assume a finite observation space, and that each observation-assignment pair $(o^i,j)$ is visited infinitely often under a constant policy $\pi$. Then, if $Q^{n+1}_i\gets FQ^{n}_i$, $\lim_{n\to\infty}Q^n_i = Q^{\pi}_i$.
\end{restatable}
\begin{proof}
The REDA target update is analogous to a SARSA update \cite{singh2000convergence}, so this can be easily proven by showing that $F$ is a $\gamma$-contraction on the space of $Q$-functions, and that $Q_i^\pi$ is the unique fixed point of this contraction. See the supplemental materials for the complete proof.
\end{proof}

The critical assumption inherent in Lemma \ref{lemma:contraction} is that the policy does not change before each observation-assignment pair is visited infinitely many times. However, assuming that a sufficiently small learning rate is chosen, Lemma \ref{lemma:contraction} will approximately hold for REDA and it can be assumed that the $Q$-values used in the mechanism $\alpha$ will converge to their desired values, $Q^\pi_i$.

This has several important implications. First, in situations where agents are providing information to centralized mechanisms for assignment, one often has to consider incentive-compatibility and whether agents are being truthful about the information they provide. Lemma \ref{lemma:contraction} shows that based on REDA's training process, the values $Q_i(o^i,j)$ used in the assignment mechanism $\alpha$ will indeed converge to $Q_i^\pi(o^i,j)$ as desired, and that agents are not able to act strategically or lie for personal benefit, a claim which we verify by experiment.

Second, it motivates that REDA is a method of approximating DQN on the joint $Q$-function $Q^\pialpha$. Because Lemma \ref{lemma:contraction} states that given enough updates, $Q_i \to Q_i^\pialpha$, Equation \ref{eqn:reda_policy} holds and $\alpha(\bfQ^\pialpha)\approx \argmax_{x^*\in X} Q^\pialpha(s,x^*)$. Then, acting according to $x_{k}=\alpha(\bfQ^\pialpha_k + \boldsymbol{\xi})$ is approximately $\epsilon$-greedy action selection with respect to $Q^\pialpha$, and target updates $\sum_{i=1}^n y^i_k=\sum_{i=1}^n r^i_k+\alpha_i(\bfQ^\pialpha_{k+1})$ approximate an optimal target update $y_k = r_k+\max_{x^*\in X}Q^\pialpha(s_{k+1},x^*)$, where $\alpha_i(\beta)$ is the value provided to agent $i$ from assignment $\alpha(\beta)$. Thus, we can expect REDA to inherit similar properties relating to the convergence of $Q^\pialpha$ to $Q^{\pi^*}$.

This key insight can explain REDA's strong performance at the system level as compared to state-of-the-art RL methods which act with independent agents, as we will see in the following section.

\section{Empirical Experiments}
\label{section:experiments}
We first test REDA in a simple SAP setting to provide intuition about why it is able to outperform existing methods in the literature. Then, we scale it up, applying it to a complex satellite constellation task allocation environment with hundreds of satellites and tasks, showing the power and efficiency of this method.

\subsection{Does REDA encourage unselfish behavior?}
\begin{figure*}[ht]
\centering
\begin{minipage}{.48\textwidth}
  \centering
  \includegraphics[width=\linewidth]{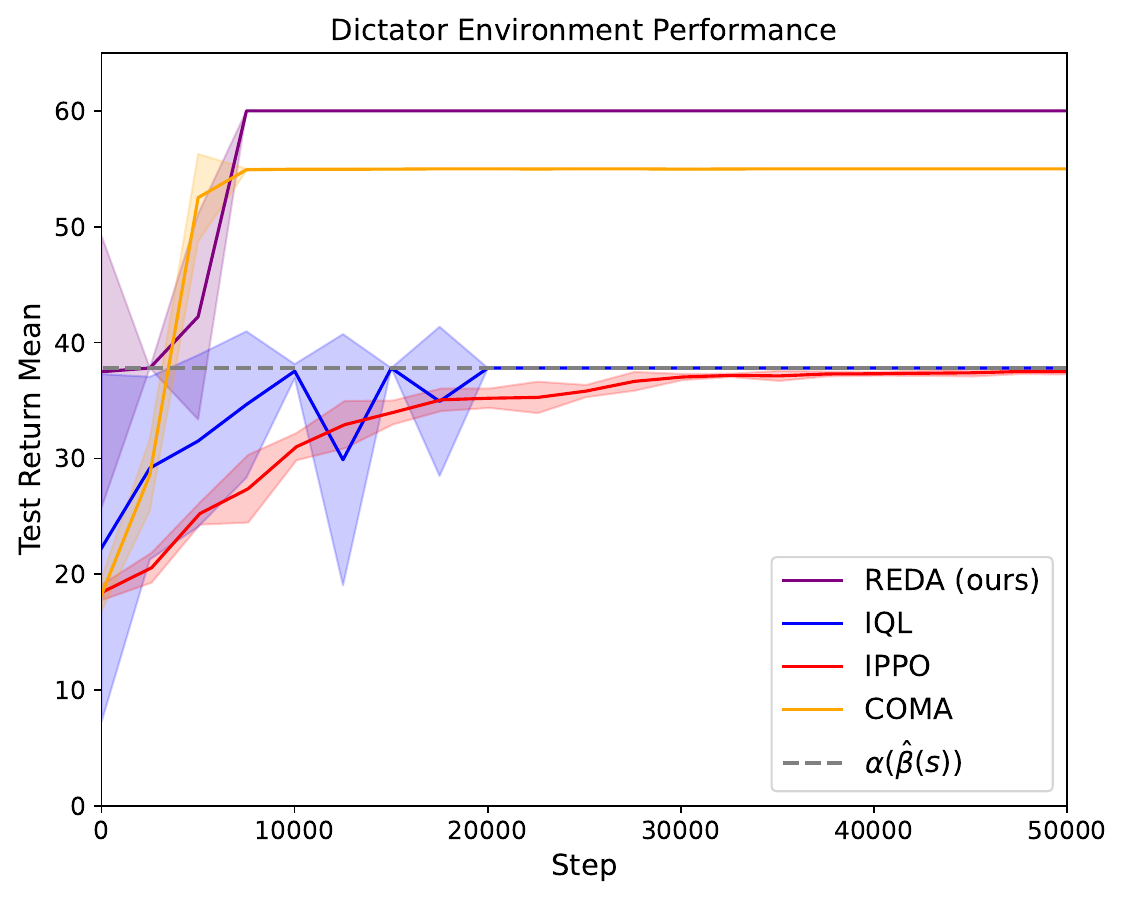}
  \captionof{figure}{Performance over 5 runs of various algorithms in dictator environment, shown with standard deviation shaded. Note that after $\epsilon$ decays to $0$ at $t=10${\small,}$000$, performance for REDA instantly approaches the theoretical maximum, while other algorithms remain significantly below the maximum.}
  \label{fig:dictator}
\end{minipage}%
\hfill
\begin{minipage}{.48\textwidth}
  \centering
  \includegraphics[width=\linewidth]{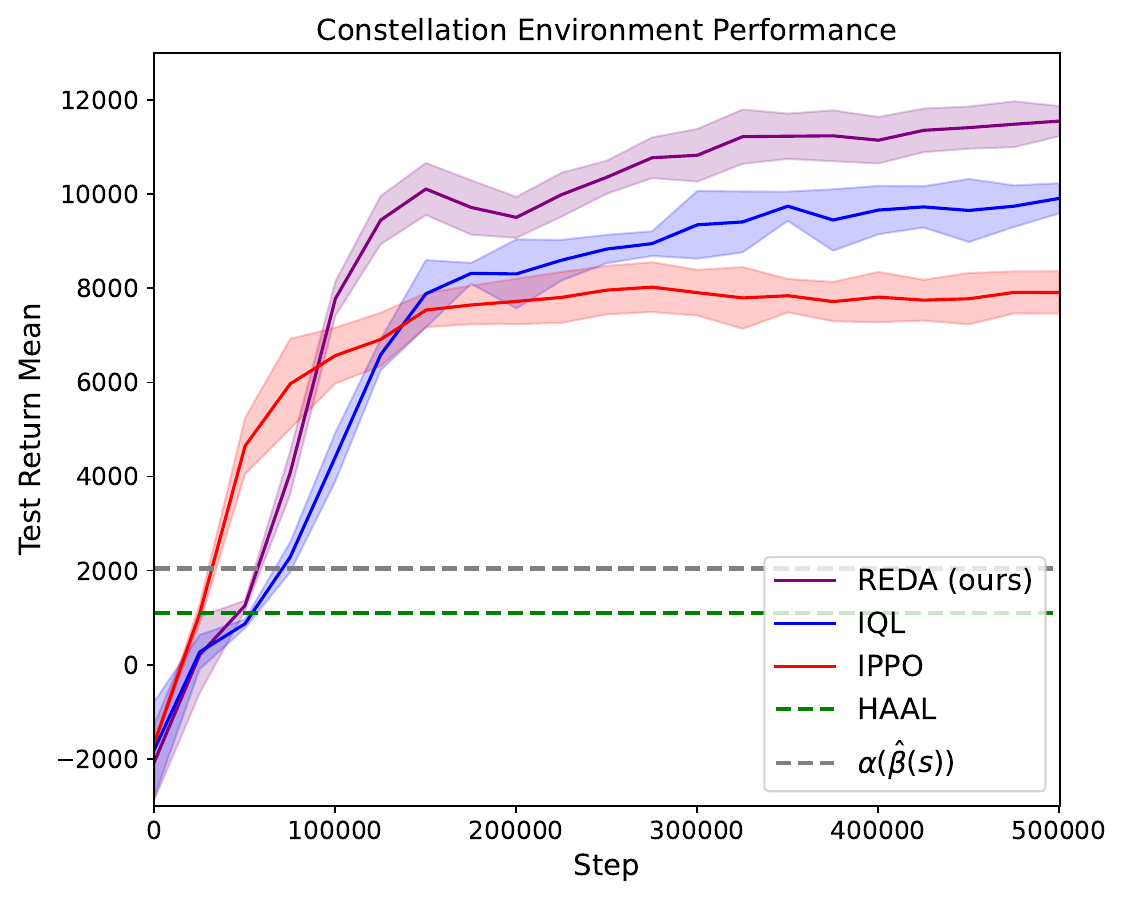}
  \captionof{figure}{Performance over $5$ runs of various algorithms in a realistic constellation environment with $324$ satellites and $450$ tasks, shown with standard deviation shaded. $\epsilon$ is decayed to zero over $300$k time steps. REDA consistently converges and obtains more reward than all other tested algorithms.}
  \label{fig:constellation_results}
\end{minipage}
\end{figure*}
We design our first experiment to test whether REDA can avoid selfish assignments. We have three agents, three tasks, and three states. Agent $1$ is the ``dictator" in that its assignment fully dictates the state transition, $s_{k+1}=x_k^1$. $\hat{\beta}(s)$ is specified as follows:
\begin{equation*}
    \hat{\beta}(1)=\begin{bmatrix}
        2 & 3 & 0 \\ 0 & 2 & 3 \\ 3 & 0 & 2
    \end{bmatrix}, \quad \hat{\beta}(2) = \begin{bmatrix}
        0 & 3 & 0 \\ 0 & 0 & 0.1 \\ 0.1 & 0 & 0
    \end{bmatrix},
\end{equation*}
\begin{equation*}
    \hat{\beta}(3) = \begin{bmatrix}
        0 & 0 & 3 \\ 0.1 & 0 & 0 \\ 0 & 0.1 & 0
    \end{bmatrix}
\end{equation*}

Rather than requiring $x\in X$, which cannot be satisfied by existing algorithms, we specify that when agents assign themselves to a task, they receive benefit corresponding to what proportion of the task they completed (i.e. $50\%$ when two agents complete the same task). This disincentivizes duplicated assignments, and means that in this SAP, the optimal policy is the joint assignment $x_k=(1, 2, 3)$ for all $k$, yielding $60$ reward over $10$ time steps. However, the greedy optimal policy for agent $1$ is to continually assign itself to task $2$, securing $3$ reward for itself each time step but causing the system as a whole to receive far less benefit by driving the state to $s=2$. 

Performance of various algorithms in this environment is shown in Figure \ref{fig:dictator}. Qualitative analysis of the results yields significant insights into the trade-offs of each method. While REDA immediately drives the group to optimal joint policy, both IQL and IPPO reliably converge to the greedy joint policy $x_k=(2,3,1) \ \forall k$. This means the dictator is acting selfishly and driving the system to state $2$ at the expense of other agents---in the satellite context, this might correspond to a satellite completing a high priority task itself rather than allowing a better suited satellite to complete it instead.

COMA avoids this greedy behavior, but still has difficulty disentangling the effect of agents on the joint reward. In experiments, it converges to the joint policy $x_k=(1,3,1) \ \forall k$, which suggests that it cannot determine whether agent $1$ or $3$ should complete task $1$.

\subsection{Can REDA be applied to large problems?}
To demonstrate the ability of REDA to scale to large complex environments, we use satellites in a constellation as agents and points on the Earth's surface as tasks. We generate a constellation of $324$ satellites evenly distributed around the Earth, with $450$ randomly placed tasks simulating internet users. Some tasks are assigned a higher priority, and thus provide higher reward for satellites who complete them. The state-dependent benefit of an assignment $i\to j$ consists of three components:
\begin{itemize}
    \item The priority of task $j$, and its distance from satellite $i$'s position in orbit.
    \item A penalty for switching assignments between two time steps (i.e. $x_{k-1}^i \neq j$), corresponding to the energy and time expenditure required for changing satellite orientation.
    \item The power state $p^i$ of the satellite. Starting at $1$ power, each time a satellite is assigned to a task in its range of visibility, $0.2$ power is expended, and being assigned to a task out of view corresponds to charging, raising power by $0.1$. Once a satellite is out of power, it captures no benefit for the rest of the episode.
\end{itemize} 

\begin{figure*}[ht!]
  \centering
  \includegraphics[width=1.5\columnwidth]{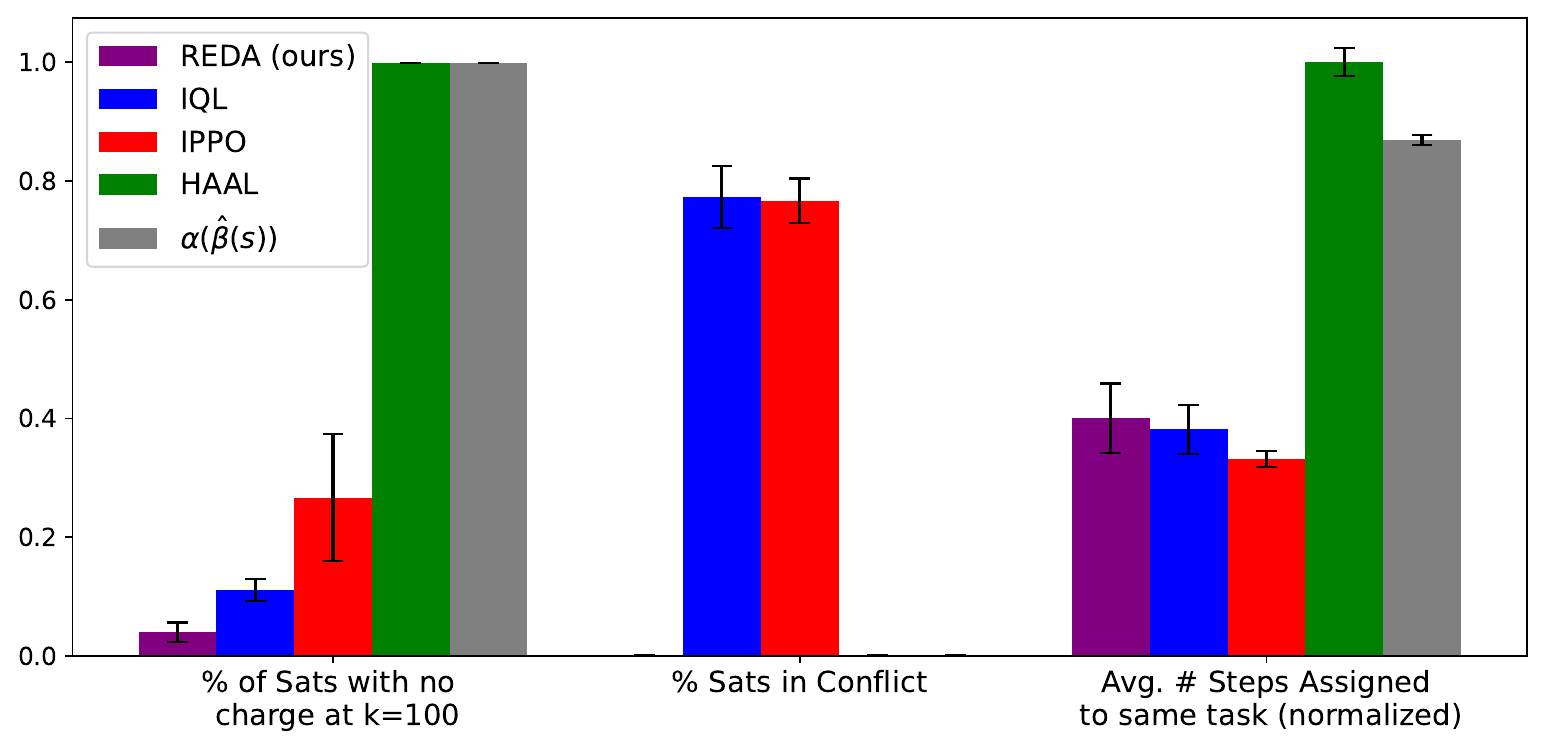}
  \caption{Performance of tested algorithms on various metrics; percentage of satellites without charge at the end of the episode (lower is better), percentage of satellites completing the same task as another satellite (lower is better), and the average number of time steps satellites are assigned to the same tasks (higher is better). We can see that REDA outperforms IQL and IPPO across the board, while avoiding having satellites run out of power as when using classical methods like HAAL.}
  \label{fig:qual_const_results}
\end{figure*}

This is a very challenging optimization problem, as satellites have to balance multiple conflicting priorities: ensuring that they are assigned to tasks that they will remain close to in their upcoming orbital motion (so as to avoid having to change assignments frequently), equitably distributing tasks amongst the hundreds of satellites, and managing their power state over a $100$ time step episode.

In order to scale to this large problem, and given that a single satellite can only physically observe a small portion of Earth's surface at a given time, we model the environment as a Partially Observed Markov Decision Process. Specifically, observations are limited to information on the top $10$ closest tasks, the previous assignment $x_{k-1}^i$, and the power state $p^i$, as well as information related to the nearest $10$ satellites to satellite $i$ in orbit.

Similarly, in reality satellites can only complete a subset of the tasks at a given time. Thus, we limit the size of the action space to $11$, the first $10$ corresponding to an assignment to the top $10$ closest tasks, with the remaining action corresponding to completing any other task. A satellite being assigned to a task it cannot complete can be interpreted as a satellite deciding to forgo benefit at that time step and instead charge its battery. A full description of the experimental setup is provided in the supplemental materials.

% Results for COMA are not provided because the size of the state space with such a large constellation causes training to become prohibitively expensive, because COMA requires training a centralized critic that can learn the value of all possible joint assignments. This makes it impractical for problems with a large joint state space.
All tested algorithms operate with this reduced observation and action space to maintain parity. Results for COMA are not provided because COMA requires training a centralized critic that can learn the value of all possible joint assignments, making it impractical for problems with a large joint state space.

Figure \ref{fig:constellation_results} shows the results of our algorithm in this environment. We find that REDA has low variance and consistently outperforms all other tested algorithms. To ensure a fair comparison, IQL was also provided with pretraining from the greedy policy, while IPPO was trained from scratch after behavior cloning on the greedy policy was found not to be beneficial. See the appendix for further details on hyperparameter selection, network architecture, and compute requirements.

In Figure \ref{fig:qual_const_results}, we can see qualitatively why REDA performs so well on this task. State-of-the-art MARL approaches avoid running out of power, but generate assignments with a significant amount of conflict. Classical approaches like HAAL make consistent assignments, but fail to manage their power state over the course of the episode, running out of power far before the end of the episode. REDA succeeds in all three areas; it entirely eliminates conflicting assignments and minimizes changes in assignment to the extent possible while still successfully managing satellite power states over the entire episode.

\section{Discussion}
We present REDA, a MARL method for learning efficient solutions to complex, state-dependent assignment problems. Rather than allow agents to act completely independently or attempt to learn the effect of their actions on the reward of the entire group, REDA strikes a balance, allowing agents to focus on learning about their personal rewards while selecting actions according to an optimal distributed task assignment mechanism. This allows the algorithm to ensure agents act unselfishly while learning efficient solutions even in very large problem settings.

\subsection{Limitations and Future Work}\label{section:limitations}
One limitation is that assignment problems, even when including state dependence, lack the expressiveness needed for certain problem settings. For example, REDA assumes that the reward can indeed be decomposed as $r(s,x)=\sum_{i=1}^n r_i(s,x^i)$. In certain scenarios, like when satellite beams can cause frequency interference with one another, this assumption does not hold. Additionally, as presented REDA only applies to scenarios where agents can only complete a single task, and where tasks can only be completed by a single agent. Future work will investigate to what extent REDA can be applied to a broader class of problems.

However, even this somewhat more limited set of problems is still broad enough to encompass a large variety of pressing abstract and practical problems. For example, REDA can be straightforwardly applied to the multiple Traveling Salesman Problem, another hugely important combinatorial optimization problem.

We also plan to apply REDA in a variety of applied problem settings, including distributed power grid management and large-scale transportation networks. REDA could even be used as a high-level planner in a hierarchical RL setting.
% While empirically the algorithm performs well, theoretical results are thus far limited to providing intuition about convergence. However, we suspect there may be a deeper theoretical connection between this work and existing strategies like DQN which may yield stronger convergence results.

\clearpage
\newpage

\section{Funding Sources}
This material is based on work supported by the National Science Foundation Graduate Research Fellowship under grant DGE-2140004.
The research of M. Mesbahi has been supported by AFOSR grant FA9550-20-1-0053. The research of N. Jaques has been supported by the Amazon Middle Mile Products and Technology gift award.

\bibliography{aaai25}

\section{Appendix}

\subsection{Proofs}
\decompthm*
\begin{proof}
    Recall the standard definition of $Q^\pialpha(s, x)$:
    \begin{align*}
        \begin{split}
            Q^\pialpha(s,x) := \bbE^\pialpha \bigg[& r(s_k,x_k) + \sum_{t=k+1}^T \gamma^{t-k} r(s_t, x_t) \\
    &\ \bigg| \ s_k=s, x_k=x, s_{k+1}\sim\calT(s_k, x_k)\bigg]
        \end{split}
    \end{align*}

    In the assignment problem setting where $r(s,x)=\sum_{i=1}^n r_i(s,x^i)$, we have:
    \begin{align*}
        \begin{split}
            Q^\pialpha(\cdot) = \bbE^\pialpha\bigg[ & \sum_{i=1}^n r_i(s_k,x_k^i) + \sum_{t=k+1}^T \sum_{i=1}^n \gamma^{t-k} r_i(s_t, x_t^i) \\ & \ \bigg| \ s_k=s, x_k=x, s_{k+1}\sim\calT(s_k, x_k) \bigg]
        \end{split}
    \end{align*}
    
    \begin{align*}
        \begin{split}
            Q^\pialpha(\cdot) = \sum_{i=1}^n \bbE^\pialpha\bigg[& r_i(s_k,x_k^i) + \sum_{t=k+1}^T \gamma^{t-k} r_i(s_t, x_t^i)\\ & \bigg| \ s_k=s, x_k=x, s_{k+1}\sim\calT(s_k, x_k) \bigg]
        \end{split}
    \end{align*}

    When $x=\pialpha(s)$, then each term inside the summation is clearly equivalent to $Q_i^\pialpha$, and we have:
    \begin{equation*}
        Q^\pialpha(s,x)=\sum_{i=1}^n Q_i^\pialpha(s,x^i) \quad \text{for } x = \pialpha(s).
    \end{equation*}

    Finally, noting that $j\neq x^i \implies x_{ij}=0$ (an agent can only be assigned to a single task), we can say:
    \begin{equation*}
        Q^\pialpha(s,x) = \sumij Q^\pialpha_i(s, j)x_{ij} \quad \text{for } x = \pialpha(s).
    \end{equation*}
    and the proof is complete.
\end{proof}

\contractlemma*
\begin{proof}
% Formally, the definition of $Q_i^\pi(o_k^i, j)$ is as follows:
% \begin{equation}\label{eqn:obs_q_defn}
%     Q_i^\pi(o_k^i,j) = \bbE_{\{s_k | o_k^i\sim\calO^i(s_k)\}} \bigg[ r_i(s_k^j,j) + \bbE^\pi V_i^\pi(s_k^j) \bigg],
% \end{equation}
Starting from an arbitrary $Q_i, \tilde{Q}_i \in \calQ$, and with $||Q||=\underset{o,j}{\sup} \ Q(o,j)$, we have:
\begin{align*}
\begin{split}
    ||FQ_i-&F\tilde{Q}_i|| \\ &= \underset{o,j}{\sup} \ \gamma \bigg|\bbE^\pi \bigg[ Q(o_{k+1}^i, x_{k+1}^i) - \tilde{Q}(o_{k+1}^i, x_{k+1}^i) \bigg] \bigg|.
\end{split}
\end{align*}

Now, $o$ and $j$ only affect the reward in that they influence $o_{k+1}^i$ and $x_{k+1}$, so we can write:
\begin{equation*}
\begin{split}
    ||FQ_i-&F\tilde{Q}_i|| \\ &\leq \underset{o_{k+1}^i,x_{k+1}^i}{\sup} \ \gamma | Q(o_{k+1}^i, x_{k+1}^i) - \tilde{Q}(o_{k+1}^i, x_{k+1}^i) | \\ &= \gamma ||Q_i-\tilde{Q}_i||.
\end{split}
\end{equation*}

This proves the operator $F$ is a $\gamma$-contraction on $Q\in\calQ$, and thus has a unique fixed point. 

Moreover, its unique fixed point is $Q_i^\pi$. Starting from the operator $F$, we have:
\begin{equation*}
    (FQ_i^\pi)(o_k^i,j) = \bbE^\pi \bigg[ r_i(s_k,j) + \gamma Q_i^\pi(o_{k+1}^i, x_{k+1}^i) \bigg] 
\end{equation*}

We abuse notation by writing $s_k \sim o_k^i$ to mean that $s_k$ is a randomly sampled state consistent with $o_k^i$. Then, substituting the definition of $Q_i^\pi$ and expanding $\bbE^\pi$,
\begin{equation*}
\begin{split}
    (FQ_i^\pi)(o_k^i,j) = \bbE \bigg[ &r_i(s_k,j) + \sum_{t=k+1}^T \gamma^{t-k} r_i(s_t, x_t^i) \\ &\ \bigg| \ s_k\sim o_k^i, s_{k+1}\sim\calT(s_k, \pi(s_k));\pi \bigg] 
\end{split}
\end{equation*}

This is now the definition of $Q_i^\pi$, so $FQ_i^\pi=Q_i^\pi$ and per Banach fixed-point theorem, if $Q^{n+1}_i=FQ^{n}_i$, then $\lim_{n\to\infty} Q^{n+1}_i \to Q_i^\pi$.
\end{proof}

\subsection{Constellation Experiment Details}\label{appendix:constellation_details}
This appendix details the specifics of the satellite constellation setup.

\subsubsection{Satellite arrangement and baseline task benefit calculation}
The satellite constellation consists of $18$ evenly spaced, circular orbital planes with $18$ satellites each for a total of $324$ satellites (altitude of $550$ km, inclination $58^\circ$). Orbital mechanics simulations are conducted with the poliastro package.

The $450$ tasks are randomly placed on the surface of Earth between $\pm 70^\circ$ latitude. A task $j$ has a randomly selected priority $P_j = \{1, 1, 1, 5\}$. Tasks are then assigned baseline benefits according to:
\begin{equation*}
    [\beta_k]_{ij} := 
    \begin{cases}
    P_j\exp\left(-\frac{[\theta_k]_{ij}^2}{2\sigma^2}\right) & \text{if }[\theta_k]_{ij} \leq \theta_{\text{FOV}} \\
    0 & \text{otherwise}
    \end{cases}
\end{equation*}
where $[\theta_k]_{ij} \in [0^\circ,360^\circ)$ is the angle between satellite $i$ and task $j$ at time step $k$, and $\theta_{\text{FOV}}:=60^\circ$ is the maximum angle at which a satellite can observe a task. This baseline benefit is later adjusted based on the state of the system $s$.

\subsubsection{Calculation of state-dependent benefits $\hat{\beta}(s_k)$}
The full state of the system $s_k$ contains the following information:
\begin{itemize}
    \item Baseline benefits of satellite task assignments $[\beta_k]_{ij} \in \bbR$ for all $i,j$, as well as baseline benefits from the following $2$ time steps $[\beta_{k+1}]_{ij}$ and $[\beta_{k+2}]_{ij}$.
    \item Previous satellite assignments $x_{k-1} \in X$.
    \item Satellite power states $p^i_k\in\bbR$, $p^i_0=1$, for all $i$.
\end{itemize}

The state dependent benefits are calculated as follows:
\[[\hat{\beta}(s_k)]_{ij} := \begin{cases}
    [\beta_k]_{ij} & \text{if } [x_{k-1}]_{ij}=1 \text{ and } p_k^i > 0\\
    [\beta_k]_{ij} - 0.5 & \text{if } [x_{k-1}]_{ij}=0 \text{ and } [\beta_k]_{ij} > 0\\
    &\text{ and } p_k^i > 0\\
    0 & \text{otherwise } (\text{i.e. }p_k^i \leq 0)
\end{cases}\]

where case $1$ corresponds to an assignment being consistent across time steps, case $2$ corresponds to switching to a meaningful task (i.e. a task that the satellite can see), and case $3$ corresponds to the case when the satellite is out of power.

\subsubsection{Observation and ``action'' spaces}
Because a given satellite only has a small portion of the $450$ tasks in view at a time, the majority of satellite-task assignments will yield no benefit and are thus interchangeable. 

Exploiting this, we can greatly restrict the observation space to the $10$ tasks with the highest state dependent benefits over the next $3$ time steps, which we denote $\calM_k^i$:
\begin{equation*}
\begin{split}
    \calM^i_k := \{j \ : \ &\sum_{t=0}^{L-1}[\beta_{k+t}]_{ij} \\ &\text{is one of the }10 \ \text{largest of any } j=1,...,m\}.
\end{split}
\end{equation*}

In a constellation with our specifications, all tasks $j\notin \calM_k^i$ are likely not visible for satellite $i$ and thus do not need to be included in the observation.
 
Similarly, we define the $10$ satellites with the highest benefits for one of the tasks in $\calM_k^i$ over the next three time steps as the neighbors $\calN^i_k$ of satellite $i$:
\begin{equation*}
\begin{split}
    \calN^i_k := \{l \ : \ &\underset{j \in \calM_k^i}{\max} \ \sum_{t=0}^{L-1}[\beta_{k+t}]_{lj} \\ &\text{ is one of the }10 \text{ largest of any } l=1,...,n\}.
\end{split}
\end{equation*}

These are satellites which are physically nearby satellite $i$, and thus it is plausible that satellite $i$ has knowledge of them.

For a satellite $i$, the observation space is then:
\begin{itemize}
    \item The baseline benefits for the best $10$ tasks over the next three time steps for itself and all neighbors, i.e. $[\beta_{k}]_{lj}$, $[\beta_{k+1}]_{lj}$, $[\beta_{k+2}]_{lj}$ for all $j \in \calM_k^i$, $l\in \{\calN_k^i \cap i\}$.
    \item The power states of itself and all neighboring satellites, i.e. $p^l_k$ for all $l\in \{\calN_k^i \cap i\}$.
    \item The previous assignments of itself and all neighboring satellites, i.e. $x_{k-1}^l$ for all $l\in \{N_k^i \cap i\}$.
\end{itemize}
This observation space provides agents with enough information to reason about which tasks competing satellites might want to bid for and to avoid penalties for running out of power or changing assignments.

Given that agents only observe information related to the top 10 benefit tasks $j\in \calM_k^i$, we also define our agent $Q$-functions only in terms of potential assignments to the top 10 tasks, as all other tasks are yield zero benefit and are therefore interchangeable.

The ``action space'' is then $11$-dimensional, with action $1$ corresponding to the task in $\calM_k^i$ with the highest benefit, action $2$ the second highest, and so on. The interchangeable actions $j\notin\calM_k^i$ are uniformly evaluated by $Q_i^\pi(o_k^i, 11)$. Notably, these tasks are almost always out of view ($[\beta_{k}]_{ij}=0$) and thus can change freely without incurring penalties for switching assignments.

The interpretation for a satellite completing a task $j\notin\calM_k^i$ is thus that the satellite chooses to forgo benefit at that time step to instead charge its battery.

\section{Hyperparameters and Network Details}\label{appendix:hyperparams}
Code for experiments, as well as network architecture and hyperparameters are adapted from the EPyMARL framework (https://github.com/uoe-agents/epymarl) \cite{papoudakis2020benchmarking}.

Action and critic networks are simple MLPs containing two hidden layers with $64$ parameters each, and use ReLU activation functions. Networks are optimized using Adam. Hyperparameters for each algorithm are provided in Table \ref{table:hyperparams}.
\begin{table}[h]
  \caption{Hyperparameters for algorithms in satellite environment.}
  \label{table:hyperparams}
  \centering
  \begin{tabular}{llll}
    \toprule
    \cmidrule(r){1-4}
    \textbf{Hyperparameter} & \textbf{REDA} & \textbf{IQL} & \textbf{IPPO} \\
    \midrule
    Discount rate $\gamma$ & 0.99 & 0.99 & 0.99 \\
    Learning rate & $0.0005$ & $0.0005$ & $0.0003$    \\
    Batch size & 5 & 5 & 6 \\
    Buffer size & 1000 & 1000 & 6 \\
    Entropy coefficient & - & - & 0.01 \\
    Target update $\tau$ & 0.01 & 0.01 & 0.01 \\
    \bottomrule
  \end{tabular}
\end{table}

The constellation environment was memory intensive, requiring $\sim300$ Gb of storage to store the replay buffer for REDA and IQL. Computation required to solve linear assignment problems was trivial as a component of the total compute, but does grow with problem size.

Experiments were run on a Macbook Pro with a M3 chip and GPU in $\sim12$ hours per run, or on a computing cluster with a single L40 Nvidia GPU for $\sim24$ hours per run.

\end{document}